\RequirePackage{fix-cm}

\documentclass[10pt]{elsarticle}

\usepackage{verbatim,a4wide}
\usepackage{amsmath}
\usepackage[section]{placeins}
\usepackage{amssymb}
\usepackage{mathtools}

\usepackage[a4paper,left=2.2cm,right=2.2cm,top=2cm,bottom=3cm]{geometry}
\usepackage{bm}
\usepackage{graphicx}
\usepackage{epstopdf}
\usepackage{arydshln}
\usepackage{amsfonts}
\usepackage{caption}
\usepackage{float}
\usepackage{paralist}
\usepackage{enumitem}
\usepackage{subfig}
\usepackage{appendix}
\usepackage{grffile}
\usepackage{bbm}
\usepackage{indentfirst}
\usepackage{booktabs}
\usepackage{multirow}
\usepackage{rotating}
\usepackage{hyperref}
\usepackage{color}

\usepackage[cp1250]{inputenc}
\usepackage[OT4]{fontenc}
\usepackage[english]{babel}

\numberwithin{equation}{section}
\newtheorem{theorem}{Theorem}[section]

\newtheorem{algorithm}[theorem]{Algorithm}

\newdefinition{remark}[theorem]{Remark}
\newdefinition{example}[theorem]{Example}
\newdefinition{problem}[theorem]{Problem}
\newproof{proof}{Proof}

\makeatletter
\def\ps@pprintTitle{%
  \let\@oddhead\@empty
  \let\@evenhead\@empty
  \let\@oddfoot\@empty
  \let\@evenfoot\@oddfoot
}
\makeatother

\newcommand{\ra}[1]{\renewcommand{\arraystretch}{#1}}

\begin{document}

\thispagestyle{empty}

\begin{frontmatter}

\title{$G^{k,l}$-constrained multi-degree reduction of B\'{e}zier curves}

\author{Przemys{\l}aw Gospodarczyk\corref{cor}}
\ead{pgo@ii.uni.wroc.pl}
\author{Stanis{\l}aw Lewanowicz}
\ead{Stanislaw.Lewanowicz@ii.uni.wroc.pl}
\author{Pawe{\l} Wo\'{z}ny}
\ead{Pawel.Wozny@ii.uni.wroc.pl}
\cortext[cor]{Corresponding author. Fax {+}48 71 3757801}
\address{Institute of Computer Science, University of Wroc{\l}aw,
         ul.~Joliot-Curie 15, 50-383 Wroc{\l}aw, Poland}

\begin{abstract}
We present a new approach to the problem of $G^{k,l}$-constrained ($k,l \leq 3$) multi-degree reduction of B\'{e}zier curves with respect to the least squares norm. First, to minimize the least squares error, we consider two methods of determining the values of geometric continuity parameters.
One of them is based on quadratic and nonlinear programming, while the other uses some simplifying assumptions and solves a system of linear equations.
Next, for prescribed values of these parameters, we obtain control points of the multi-degree reduced curve, using the properties of
constrained dual Bernstein basis polynomials. Assuming that the input and output curves are of degree $n$ and $m$, respectively,
we determine these points with the complexity $O(mn)$, which is significantly less than the cost of other known methods.
Finally, we give several examples to demonstrate the effectiveness of our algorithms.
\end{abstract}


\end{frontmatter}

\section{Introduction}\label{Sec:Intro}
 Let $\Pi_n^d$ denote the space of all parametric polynomials in $\mathbb{R}^d$ of degree at most $n$;
$\Pi_n^1 := \Pi_n$.

A \textit{B\'{e}zier curve} $P_n \in \Pi_n^d$ of degree $n \in \mathbb{N}$ is the following parametric curve:
\begin{equation}\label{Eq:GOriginal}
    P_n(t) := \sum_{i = 0}^{n} p_{i}B_i^{n}(t) \qquad(0\leq t\leq 1),
\end{equation}
where $p_0,p_1,\ldots, p_n \in \mathbb{R}^d$ are so-called \textit{control points}, and  $B_0^n,B_1^n,\ldots,B_n^n$ are the \textit{Bernstein polynomials of degree $n$} given by
\begin{equation}\label{Eq:Bernstein}
B_i^n(t) := \binom{n}{i}t^i(1-t)^{n-i} \qquad (0 \leq i \leq n).
\end{equation}
In this paper, we consider the following problem.
\begin{problem}\,[\textsf{$G^{k,l}$-constrained multi-degree reduction}]  \label{P:GProblem}
\ \\ \textit{For a given B\'{e}zier curve $P_n$ of degree $n$,
find a B\'{e}zier curve $R_m$ of lower degree $m$,}
\begin{equation}\label{Eq:GReduced}
R_m(t) := \sum_{i=0}^{m}r_iB_i^m(t)\qquad (0\leq t\leq 1),
\end{equation}
\textit{so that the following conditions are satisfied:}
\begin{description}
\itemsep4pt
\item[(i)] \textit{$P_n$ and $R_m$ are 	$G^{k,l}$-continuous ($-1 \leq k,l \leq 3$ and $k+l<m-1$)
at the endpoints, i.e.,}
\begin{equation}\label{Eq:Constraints}
\left.\begin{array}{ll}
\displaystyle \frac{\mbox{d}^i}{\mbox{d}t^i} R_m(t) = \frac{\mbox{d}^i}{\mbox{d}t^i} P_n(\varphi(t)) &\quad (t = 0; \ i = 0, 1,\dots, k),\\[1ex]
\displaystyle \frac{\mbox{d}^j}{\mbox{d}t^j} R_m(t) = \frac{\mbox{d}^j}{\mbox{d}t^j} P_n(\varphi(t)) &\quad (t = 1; \ j =
0, 1,\dots, l),
\end{array}\right\}
\end{equation}
\textit{where $\varphi:[0,1]\to[0,1]$ is a strictly increasing function with $\varphi(0) = 0$ and $\varphi(1) = 1$;}

\item[(ii)] \textit{value of the squared $L_2$-error}
	\begin{equation*}
 ||P_n-R_m||_{L_2}^2 := \int_0^1 (1-t)^{\alpha}t^{\beta}||P_n(t)-R_m(t)||^2 \mbox{d}t
\qquad (\alpha, \ \beta > -1)
\end{equation*}
\textit{is minimized in the space $\Pi_m^d$, where $||\cdot||$ is the Euclidean vector norm.}
\end{description}
\end{problem}
Problems of the above type have been recently discussed in several papers~\cite{Lu13,LW06,WL07,RM11,RM13,ZW10,ZWY13},
usually under simplifying assumptions $\varphi'(0)=\varphi'(1)=1$, which implied, for example, the hybrid
$C^{1,1}/G^{2,2}$-constrained degree reduction, meaning that we impose constraints of
$C^{1,1}$-continuity, followed by $G^{2,2}$-continuity, at the endpoints.
Most of the known algorithms solve a system of normal equations to get control points of the multi-degree reduced curve~\eqref{Eq:GReduced}.
Consequently, solution depends on the inverse of a certain matrix, so the obtained formulas are not truly explicit and the cost of the method is high (see, e.g.,~\cite{Lu13,WL07,RM13}).
For extensive lists of references, see the recent papers of Lu~\cite{Lu13}, or Rababah and Mann~\cite{RM13}.
The conventional problem of degree reduction differs from Problem~\ref{P:GProblem} in considering, instead of condition (i), the $C^{k,l}$-continuity at the endpoints of curves, i.e.,
\begin{equation}\label{Eq:CConstraints}
\left.\begin{array}{ll}
R^{(i)}_m(0) = P^{(i)}_n(0) \qquad (i = 0, 1,\dots, k),\\
R^{(j)}_m(1) = P^{(j)}_n(1) \qquad (j = 0, 1,\dots, l).
\end{array}\right\}
\end{equation}
In the past $30$ years, many papers dealing with this problem have been published (see, e.g.,~\cite{CW02,Eck95,Sun05,SL04,WL09}).
In particular,  in~\cite{WL09},  two of us have proposed a method based on the use of the so-called \textit{dual Bernstein polynomials}, which has complexity $O(mn)$, the least among the existing algorithms.
In the present paper, we apply an extended version of this method as an essential part of the algorithms
of solving Problem~\ref{P:GProblem}. Such an approach allows us to avoid matrix inversion.
Assuming that $-1 \leq k,l \leq 3$ and including the hybrid cases, there are $37$ continuity cases which require computation of the continuity parameters. Those variants of the problem differ, and we have not proven that in each case a unique solution exists.

The outline of the paper is as follows. Section~\ref{Sec:Pre} contains a preliminary material. In Section~\ref{Sec:Geo},
we relate the $G^{k,l}$-continuity conditions with the control points of the curves $P_n$ and $R_m$. Section~\ref{Sec:Prob}
brings complete solutions of Problem~\ref{P:GProblem}, with and without the simplifying assumptions.
Section~\ref{Sec:Alg} deals with algorithmic implementation of the proposed methods. In Section~\ref{Sec:Ex}, we give some examples showing efficiency of our methods. Conclusions are given in Section~\ref{Sec:Conc}.

\section{Preliminaries}\label{Sec:Pre}

In this section, we introduce necessary definitions and notation.

We define the inner product $\langle\cdot,\cdot\rangle_{\alpha,\beta}$ by
\begin{equation}\label{Eq:inprod}
\langle f,g\rangle_{\alpha,\beta} := \int_0^1 (1-t)^{\alpha}t^{\beta}f(t)g(t) \mbox{d}t \qquad (\alpha, \ \beta > -1).
\end{equation}
There is a unique \textit{dual Bernstein polynomial basis of degree $n$}
\begin{equation*}
D_0^{n},D_1^{n},\ldots, D_n^{n} \in \Pi_n,
\end{equation*}
associated with the basis~\eqref{Eq:Bernstein}, so that
\begin{equation*}
\left\langle D_i^{n},B_j^n\right\rangle_{\alpha,\beta} =\delta_{ij} \qquad (i,j=0,1,\ldots, n),
\end{equation*}
where $\delta_{ij}$ equals $1$ if $i=j$, and $0$ otherwise.

Given the integers $k, l$ such that $k, l \geq -1$ and $k+l<n-1$, let $\Pi_n^{(k,l)}$ be the space of all
polynomials of degree at most $n$, whose derivatives of orders $0,1,\ldots,k$ at $t=0$, as well as
derivatives of orders $0,1,\ldots,l$ at $t=1$, vanish. We use the convention that derivative of order $0$ of a function is the function itself.
Clearly, $\mbox{dim}$ $\Pi_n^{(k,l)} = n-k-l-1$,
and the Bernstein polynomials $\left\{B_{k+1}^{n},B_{k+2}^{n},\ldots, B_{n-l-1}^{n}\right\}$
form a basis of this space. There is a unique \textit{dual constrained Bernstein polynomial basis of degree $n$}
\begin{equation*}
\left\{D_{k+1}^{(n,k,l)},D_{k+2}^{(n,k,l)},\ldots, D_{n-l-1}^{(n,k,l)}\right\}
\subset \Pi_n^{(k,l)}
\end{equation*}
satisfying the relation $\left\langle D_i^{(n,k,l)},B_j^n\right\rangle_{\alpha,\beta} =\delta_{ij} \ (i,j=k+1,k+2,\ldots, n-l-1)$.
Obviously, we have $D_i^{(n,-1,-1)}=D_i^{n}$, which corresponds to the case without any constraints. For properties of the polynomials $B_i^n$ and $D_i^{n}$, see~\cite[Appendix A]{WL09}.

\textit{Forward difference operator} is given by
$$
    \Delta^0q_i := q_i,\quad
    \Delta^kq_i := \Delta^{k-1}q_{i+1} - \Delta^{k-1}q_{i} \quad(k =1,2,\ldots).
$$

We use $C^{p,q}/G^{k,l}$ notation to describe the hybrid constraints, where $p,q \in \{-,1\}$ and
($k \geq 2$ or $l \geq 2$). In the case of $k \geq 2$ and $p = 1$, we set $\varphi'(0) := 1$.
Similarly, for $l \geq 2$ and $q = 1$, we set $\varphi'(1) := 1$. Setting $p := -$, $q := -$ means that we do not fix
$\varphi'(0)$, $\varphi'(1)$, respectively. Clearly, $C^{-,-}/G^{k,l}$ denotes $G^{k,l}$.

\section{Geometric continuity}\label{Sec:Geo}

In this section, we relate the $G^{k,l}$-continuity conditions~\eqref{Eq:Constraints} with the
control points. We limit ourselves to $k,l \leq 3$ cases, which are the most important from a practical point of view.

Remark that the control points $r_1,\,r_2,\ldots,\,r_k$  depend on the parameters
$$
\lambda_j := \varphi^{(j)}(0) \qquad (j=1,2,\ldots,k),
$$
while the points $r_{m-1},r_{m-2},\ldots,r_{m-l}$ depend on
$$
\mu_j := \varphi^{(j)}(1) \qquad (j=1,2,\ldots,l).
$$

Now, let us recall the well known formulas (see \cite{Lu14}, also \cite{Lu13,RM11}). When $k = 3$, we have:
\begin{align}
\displaystyle & r_0 =  p_0, \quad r_1 = p_0 + \frac{n}{m}\lambda_1\Delta p_0,\label{Eq:GeoCondsGk1}\\[2ex]
\displaystyle & r_2 =  p_0 + \frac{n}{m}\left[2\lambda_1 + \frac{1}{m-1}\lambda_2\right]\Delta p_0 + \frac{(n-1)_2}{(m-1)_2}\lambda_1^2\Delta^2p_0,\label{Eq:GeoCondsGk2}\\[2ex]
\displaystyle & r_3 =  p_0 + \frac{n}{m}\left[3\lambda_1 + \frac{3}{m-1}\lambda_2  + \frac{1}{(m-2)_2}\lambda_3\right]\Delta p_0 \nonumber \\
& \hphantom{r_2 =  } + 3\frac{(n-1)_2}{(m-1)_2}\left[\lambda_1^2 + \frac{1}{m-2}\lambda_1\lambda_2\right]\Delta^2p_0 + \frac{(n-2)_3}{(m-2)_3}\lambda_1^3\Delta^3 p_0.\label{Eq:GeoCondsGk3}
\end{align}
In the case of $k = 2$, we use~\eqref{Eq:GeoCondsGk1}~and~\eqref{Eq:GeoCondsGk2}. For $k = 1$, formulas~\eqref{Eq:GeoCondsGk1} hold.
Analogously, when $l = 3$, we have:
\begin{align}
\displaystyle & r_m = p_n, \quad r_{m-1} = p_n - \frac{n}{m}\mu_1\Delta p_{n-1},\label{Eq:GeoCondsGl1}\\[2ex]
\displaystyle & r_{m-2} = p_n - \frac{n}{m}\left[2\mu_1 - \frac{1}{m-1}\mu_2\right]\Delta p_{n-1} + \frac{(n-1)_2}{(m-1)_2}\mu_1^2\Delta^2 p_{n-2},\label{Eq:GeoCondsGl2}\\[2ex]
\displaystyle & r_{m-3} = p_n - \frac{n}{m}\left[3\mu_1 - \frac{3}{m-1}\mu_2 + \frac{1}{(m-2)_2}\mu_3\right]\Delta p_{n-1} \nonumber \\
& \hphantom{r_{m-3} =  } + 3\frac{(n-1)_2}{(m-1)_2}\left[\mu_1^2 - \frac{1}{m-2}\mu_1\mu_2\right]\Delta^2p_{n-2} - \frac{(n-2)_3}{(m-2)_3}\mu_1^3\Delta^3 p_{n-3}.\label{Eq:GeoCondsGl3}
\end{align}
In the case of $l = 2$, we use~\eqref{Eq:GeoCondsGl1}~and~\eqref{Eq:GeoCondsGl2}. For $l = 1$, formulas~\eqref{Eq:GeoCondsGl1} hold.

\section{$G^{k,l}$-constrained multi-degree reduction problem}\label{Sec:Prob}

\subsection{Multi-degree reduction of B\'ezier curves with prescribed boundary control points}\label{Subsec:Prob0}

First, we discuss the following \textit{model problem} of constrained multi-degree reduction:

\begin{problem}\,[\textsf{Multi-degree reduction with prescribed boundary control points}] \label{P:Problem0}
\ \\
\textit{Given a B\'{e}zier curve $P_n \in \Pi_n^d$,}
\begin{equation*}\label{Eq:Original0}
P_n(t) := \sum_{i=0}^{n}p_iB_i^n(t),
\end{equation*}
\textit{we look for a B\'{e}zier curve $R_m \in \Pi_m^d \ (m < n)$,}
\begin{equation}\label{Eq:Reduced0}
R_m(t) := \sum_{i=0}^{m}r_iB_i^m(t),
\end{equation}
\textit{having the prescribed control points $r_0,r_1,\ldots,r_k$ and $r_{m-l},r_{m-l+1},\ldots,r_{m}$,
that gives minimum value of the error}
\begin{equation}\label{Eq:Distance0}
E^{(\alpha,\beta)} := ||P_n-R_m||_{L_2}^2 = \int_0^1 (1-t)^{\alpha}t^{\beta}||P_n(t)-R_m(t)||^2 \mbox{d}t \qquad (\alpha, \ \beta > -1).
\end{equation}
\end{problem}
Given the points $p_i:=(p_{i1},p_{i2}\ldots,p_{id}) \in \mathbb{R}^d$ ($i=0,1,\ldots,n$) and $r_i:=(r_{i1},r_{i2}\ldots,r_{id}) \in \mathbb{R}^d$ ($i=0,1,\ldots,m$), we use notation $\mathbf{p}^h$, $\mathbf{r}^h$ for the vectors of $h$th coordinates of the points $p_0,p_1,\ldots,p_n$ and $r_0,r_1,\ldots,r_m$, respectively:
\[
	\mathbf{p}^h:=[p_{0h},p_{1h},\ldots, p_{nh}],\qquad\mathbf{r}^h:=[r_{0h},r_{1h},\ldots, r_{mh}]
	\qquad(h=1,2,\ldots,d).		
	\]
As an extension of the result given in~\cite{WL09} (see also~\cite{LW11}), we obtain the following theorem.

\begin{theorem}\label{Th:Problem0} 
	The \emph{inner} control points
$r_i=(r_{i1},r_{i2},\ldots,r_{id}) \ (k+1 \leq i \leq m-l-1)$
of the curve~\eqref{Eq:Reduced0}, being the solution of the Problem~\ref{P:Problem0}, are given by
\begin{eqnarray}\label{Eq:Solution0}
	r_i=\sum_{j=0}^{n}\upsilon_j\phi_{ij}
     \qquad (i = k+1,k+2,\ldots, m-l-1),
\end{eqnarray}
where
\begin{equation}
	\label{Eq:phi}
	\phi_{ij}:=\left\langle B_j^n,D_i^{(m,k,l)}\right\rangle_{\alpha,\beta},	
	\end{equation}
and
\begin{align}\label{Eq:Coeffsups}
	\upsilon_j :=& p_j -\binom{n}{j}^{-1} \left(\sum_{h=0}^{k}+\sum_{h=m-l}^{m}\right) \binom{n-m}{j-h}\binom{m}{h}
	r_{h} \qquad (j = 0,1,\ldots, n).
\end{align}

The squared $L_2$-error~\eqref{Eq:Distance0} is given by
\begin{equation}\label{Eq:Distance2}
E^{(\alpha,\beta)} = \sum_{h=1}^{d}\left\{I_{nn}(\mathbf{p}^h,\mathbf{p}^h) + I_{mm}(\mathbf{r}^h,\mathbf{r}^h) - 2I_{nm}(\mathbf{p}^h,\mathbf{r}^h)\right\},
\end{equation}
where for $\mathbf{a}:=[a_0,a_1,\ldots, a_N]$ and $\mathbf{b}:=[b_0,b_1,\ldots, b_M]$, we define
$$
I_{NM}(\mathbf{a},\mathbf{b}) := \frac{\mathrm{B}(\alpha+1,\beta+1)}{(\alpha+\beta+2)_{N+M}}
\sum_{i=0}^{N}\sum_{j=0}^{M}\binom{N}{i}\binom{M}{j}(\alpha+1)_{N+M-i-j}(\beta+1)_{i+j}a_ib_j,
$$
where $\mathrm{B}(\alpha,\beta):=\frac{\Gamma(\alpha)\Gamma(\beta)}{\Gamma(\alpha+\beta)}$ is the beta function.
\end{theorem}

\begin{proof}
Let us write
$$
R_m(t) = S_m(t) + T_m(t),
$$
where
$$
S_m(t) := \sum_{i=k+1}^{m-l-1}r_iB_i^m(t), \quad T_m(t) := \left(\sum_{i=0}^{k}+\sum_{i=m-l}^{m}\right)r_iB_i^m(t).
$$
Using the degree elevation formula (see, e.g.,~\cite[\S6.10]{Far02};  we adopt the usual convention that $\binom{u}{v}=0$ if $v<0$ or $v>u$)
$$
B_i^m(t) =\binom{m}{i} \sum_{h=0}^{n} \binom{n-m}{h-i}\binom{n}{h}^{-1}B_{h}^{n}(t),
$$
we write
$$
T_m(t) = \sum_{j=0}^{n}d_jB_j^n(t),
$$
where
$$
d_j := \binom{n}{j}^{-1}\left(\sum_{h=0}^{k}+\sum_{h=m-l}^{m}\right)
\binom{n-m}{j-h}\binom{m}{h}r_h.
$$

Now, we observe that
$$
||P_n-R_m||_{L_2}^2 = ||W_n-S_m||_{L_2}^2 = \sum_{h=1}^d\int_0^1(1-t)^{\alpha}t^{\beta}\left[W^h_n(t)-S^h_m(t)\right]^2\mbox{d}t,
$$
where
\begin{align*}
&W_n(t) := \left[W_n^1(t),W_n^2(t),\ldots,W_n^d(t)\right] = P_n(t) - T_m(t) = \sum_{i=0}^{n} \upsilon_iB_i^n(t),\\
&S_m(t) := \left[S_m^1(t),S_m^2(t),\ldots,S_m^d(t)\right],
\end{align*}
with
$$
\upsilon_i := p_i-d_i.
$$
Thus, we are looking for the best least squares approximation for $W_n^h \ (h = 1,2,\ldots,d)$ in the space $\Pi_m^{(k,l)}$.
Remembering that $B_i^m$ and $D_i^{(m,k,l)}\ (k+1 \leq i \leq m-l-1)$ are the dual bases in the space
$\Pi_m^{(k,l)}$, we obtain
\begin{equation*}
r_i = \sum_{j=0}^{n}\upsilon_j\left\langle B_j^n,D_i^{(m,k,l)}\right\rangle_{\alpha,\beta}=\sum_{j=0}^{n}\upsilon_j\phi_{ij} \qquad (i = k+1,k+2,\ldots, m-l-1),
\end{equation*}
which is  the formula~\eqref{Eq:Solution0}.

Proof of~\eqref{Eq:Distance2} uses an argument similar to the one given in~\cite{WL09}.\qed
\end{proof}

\begin{remark}\label{R:AuxAlg}
Let us define the quantities $\psi_{ij}$  ($i=k+1,k+2,\ldots,m-l-1;\ j =0,1,\ldots,n$), related to the coefficients $\phi_{ij}$ (cf.~\eqref{Eq:phi}) by  the following formula:
\begin{equation}\label{Eq:psi}
	\phi_{ij} := \binom{m-k-l-2}{i-k-1}\binom{m}{i}^{-1}\binom{n}{j}
	\frac{(\alpha+l+2)_{n-j}(\beta+k+2)_{j}}{(\alpha+l+2)_{l+1}(\beta+k+2)_{k+1}}\psi_{ij}.
\end{equation}
Observe that the quantities $\psi_{ij}$ can be put in a rectangular table and the entries of this $\psi$-table can be computed
using \cite[Algorithm 4.2]{LW11}, assuming that $c_1 := k+1$, $c_2 := l+1$, $\alpha_1 := \alpha$ and $\alpha_2 := \beta$.
Note that the complexity of this algorithm is $O(mn)$.
\end{remark}

\subsection{$G^{k,l}$-constrained multi-degree reduction}\label{Subsec:Prob1}

Coming back to the problem of $G^{k,l}$-constrained multi-degree reduction (see~Problem~\ref{P:GProblem}), let us
notice that the formulas~\eqref{Eq:GeoCondsGk1}--\eqref{Eq:GeoCondsGl3} with \textit{fixed parameters} $\lambda_i$ and $\mu_j$
(cf.~\S\ref{Sec:Geo}) constitute constraints of the form demanded in Problem~\ref{P:Problem0}. As a result,
 the  control points~\eqref{Eq:Solution0} depend on these parameters.

Now,  the \textit{optimum values} of the parameters can be obtained by minimizing
the error function~\eqref{Eq:Distance2},
\begin{equation}\label{Eq:E2fun}
	E^{(\alpha,\beta)} \equiv E^{(\alpha,\beta)}(\lambda_1,\lambda_2,\ldots,\lambda_k,\mu_1,\mu_2,\ldots,\mu_l),
\end{equation}	
	depending on $\{\lambda_i\}$ and $\{\mu_j\}$ \textit{via} formulas~\eqref{Eq:GeoCondsGk1}--\eqref{Eq:GeoCondsGl3} and~\eqref{Eq:Solution0}.
	
	For a minimum of function~\eqref{Eq:E2fun}, it is necessary that the derivatives of $E^{(\alpha,\beta)}$
 with respect to the parameters are zero, which yields the system
\begin{equation}\label{Eq:GeoSystem-mod}
	\begin{array}{l}
         \displaystyle        \sum_{h=1}^{d}\sum_{j=u}^{m-l-1}\left[F_{mj}(\mathbf{r}^h)-F_{nj}(\mathbf{p}^h)
                                        \right]\frac{\partial r_{jh}}{\partial \lambda_u}
                                        = 0 \qquad (u=1,2,\ldots, k),\\[3ex]
                    \displaystyle \sum_{h=1}^{d}\sum_{j=k+1}^{m-v}\left[F_{mj}(\mathbf{r}^h)-F_{nj}(\mathbf{p}^h)
                                        \right]\frac{\partial r_{jh}}{\partial \mu_v}
                                        = 0 \qquad (v=1,2,\ldots, l),
                                   \end{array}
                                   \end{equation}
where we use notation
\begin{equation*}
F_{tj}(\mathbf{q}):=\frac{1}{(\alpha+\beta+m+2)_{t}}\binom{m}{j}\sum_{i=0}^{t}\binom{t}{i}
(\alpha+1)_{t+m-i-j}(\beta+1)_{i+j}q_{i}
\end{equation*}
with $\mathbf{q}=[q_0,q_1,\ldots,q_t]$.

In the case of $k = l = 3$, we compute the partial derivatives of $h$th coordinates of the control points~\eqref{Eq:GeoCondsGk1}--\eqref{Eq:GeoCondsGl3}. We obtain:
\begin{align}
	\label{Eq:u1}
\displaystyle \frac{\partial r_{ih}}{\partial \lambda_1} =&
\left\{\begin{array}{ll}
	  \frac{n}{m}\Delta p_{0h}&\quad (i=1),  \\[2ex]
	  2\frac{n}{m}\Delta p_{0h}+ 2\lambda_1\frac{(n-1)_2}{(m-1)_2}\Delta^2 p_{0h}&\quad (i=2),\\[2ex]
	  3\frac{n}{m}\Delta p_{0h}  + \left[2\lambda_1
	                              + \frac{1}{m-2}\lambda_2\right]3\frac{(n-1)_2}{(m-1)_2}\Delta^2 p_{0h}
	  + 3\lambda_1^2\frac{(n-2)_3}{(m-2)_3}\Delta^3 p_{0h}  &\quad (i=3),\\[2ex]
	  0 &\quad(i = 0;\: m-3\le i\le m),
\end{array}\right.\\
	\label{Eq:u2}
\displaystyle \frac{\partial r_{ih}}{\partial \lambda_2} =&
\left\{\begin{array}{ll}
    \frac{n}{(m-1)_2}\Delta p_{0h}&\quad (i=2),\\[2ex]
    3\frac{n}{(m-1)_2}\Delta p_{0h}
                     + 3\lambda_1\frac{(n-1)_2}{(m-2)_3}\Delta^2 p_{0h}&\quad (i=3),\\[2ex]
    0 &\quad (i = 0,1; \: m-3\le i\le m),
\end{array}\right.\\
	\label{Eq:u3}
\displaystyle \frac{\partial r_{ih}}{\partial \lambda_3} =&
\left\{\begin{array}{ll}
    \frac{n}{(m-2)_3}\Delta p_{0h}&\quad (i=3),\\[2ex]
    0 &\quad (i = 0,1,2; \: m-3\le i\le m),
\end{array}
\right.
\end{align}

\begin{align}
	\label{Eq:v1}
\displaystyle \frac{\partial r_{ih}}{\partial \mu_1} =&
\left\{\begin{array}{ll}
	  -\frac{n}{m}\Delta p_{n-1,h}&\quad (i=m-1),  \\[2ex]
	  -2\frac{n}{m}\Delta p_{n-1,h}
             + 2\mu_1\frac{(n-1)_2}{(m-1)_2}\Delta^2 p_{n-2,h}&\quad (i=m-2),\\[2ex]
	  -3\frac{n}{m}\Delta p_{n-1,h}  + \left[2\mu_1
	  - \frac{1}{m-2}\mu_2\right]3\frac{(n-1)_2}{(m-1)_2}\Delta^2 p_{n-2,h}
	  - 3 \mu_1^2\frac{(n-2)_3}{(m-2)_3}\Delta^3 p_{n-3,h}
	  &\quad (i=m-3),\\[2ex]
	  0 &\quad(0\le i\le3;\: i=m),
\end{array}\right.\\
	\label{Eq:v2}
\displaystyle \frac{\partial r_{ih}}{\partial \mu_2} =&
\left\{\begin{array}{ll}	
	  \frac{n}{(m-1)_2}\Delta p_{n-1,h}&\quad (i=m-2),\\[2ex]
	  3\frac{n}{(m-1)_2}\Delta p_{n-1,h}
                     - 3\mu_1\frac{(n-1)_2}{(m-2)_3}\Delta^2 p_{n-2,h}
	  &\quad (i=m-3),\\[2ex]
	  0 &\quad(0\le i\le3;\: i=m-1,m),
\end{array}\right.\\
	\label{Eq:v3}
\displaystyle \frac{\partial r_{ih}}{\partial \mu_3} =&
\left\{\begin{array}{ll}	 	
	  -\frac{n}{(m-2)_3}\Delta p_{n-1,h}&\quad (i=m-3),\\[2ex]
	  0 &\quad(0\le i\le3;\: m-2\le i\le m).
\end{array}
\right.
\end{align}

Notice that the partial derivatives of $h$th coordinates of control points~\eqref{Eq:Solution0}
depend on~\eqref{Eq:u1}--\eqref{Eq:v3} in the following way:
\begin{align}
\displaystyle & \frac{\partial r_{ih}}{\partial \lambda_u} = -\sum_{j=0}^n \binom{n}{j}^{-1}\sum_{g=u}^k
\binom{n-m}{j-g}\binom{m}{g}\phi_{ij}\frac{\partial r_{gh}}{\partial \lambda_u},\label{Eq:InnerDerLambda} \\[3ex]
\displaystyle & \frac{\partial r_{ih}}{\partial \mu_v} = -\sum_{j=0}^n \binom{n}{j}^{-1}\sum_{g=m-l}^{m-v}
\binom{n-m}{j-g}\binom{m}{g}\phi_{ij}\frac{\partial r_{gh}}{\partial \mu_v}. \label{Eq:InnerDerMu}
\end{align}

One can easily see, that when $k, l \leq 3$, we compute $\frac{\partial r_{ih}}{\partial \lambda_u}$, $\frac{\partial r_{ih}}{\partial \mu_v}$ by~\eqref{Eq:InnerDerLambda},~\eqref{Eq:InnerDerMu} if $k<i<m-l$, and by~\eqref{Eq:u1}--\eqref{Eq:v3} otherwise.
Finally, we put the expressions~\eqref{Eq:u1}--\eqref{Eq:InnerDerMu} into the equations of system~\eqref{Eq:GeoSystem-mod}.

Observe that for $k \geq 2$ or $l \geq 2$, system~\eqref{Eq:GeoSystem-mod} is nonlinear,
which makes it quite difficult to solve. Furthermore, from a practical point of view, we additionally require
that $\lambda_1, \mu_1 > 0$, which results in the same directions of tangent vectors at the endpoints of curves~\eqref{Eq:GOriginal} and~\eqref{Eq:GReduced}.
Therefore, to guarantee that these conditions will be satisfied, it is not enough just to solve the system~\eqref{Eq:GeoSystem-mod}.

Now, let us discuss two possible ways of determining the values of geometric continuity parameters.

\subsubsection{Determining the  $G^{k,l}$ parameters using optimization methods}\label{SubSec:NLPGeo}

It is easy to check that if ($k = 1$ and $l \leq 1$) or ($l = 1$ and $k \leq 1$),
then the error~\eqref{Eq:Distance2} is a~quadratic function of
continuity parameters.

In the case of ($k = 2$ and $l \leq 2$) or ($l = 2$ and $k \leq 2$),
the error~\eqref{Eq:Distance2} is a fourth-degree polynomial function of
continuity parameters.

For ($k = 3$ and $l \leq 3$) or ($l = 3$ and $k \leq 3$),
the error~\eqref{Eq:Distance2} is a sixth-degree polynomial function of
continuity parameters.

To find the optimum values of parameters $\lambda_1$, $\mu_1$ in the case of $G^{1,1}$-constrained
multi-degree reduction problem, assuming that $\alpha, \beta = 0$, Lu and Wang~\cite{WL07}
solve the quadratic programming problem, subject to the constraints
\begin{equation}\label{Eq:PositiveCons}
\lambda_1 \geq d_0, \quad  \mu_1 \geq d_1,
\end{equation}
where $d_0$ and $d_1$ are positive lower bounds, usually prescribed to small values
(they set $10^{-4}$ for both lower bounds in the examples section).
Such approach can be used in the cases which result in the quadratic error function~\eqref{Eq:Distance2}.
One can solve the quadratic programming problem using, e.g., an \textit{iterative active-set method}, which is implemented in many software libraries.
The active-set mechanism used by standard quadratic solvers is described in~\cite[\S6.5]{BG97}.

Analogously, one can observe that for $k=2,3$ or $l=2,3$,
the problem of minimizing the error~\eqref{Eq:Distance2}, subject to the constraints~\eqref{Eq:PositiveCons} is a nonlinear programming problem. To solve it, one can use, for instance, a \textit{sequential quadratic programming (SQP)
method} (see, e.g.,~\cite[\S15.1]{BG97}), which is available in many software libraries.

\subsubsection{Determining the $C^{p,q}/G^{k,l}$ parameters by solving a system of linear equations}\label{SubSec:CGeo}

In the case of $G^{2,2}$, Rababah and Mann~\cite{RM11} simplified the problem by considering $C^{1,1}$-continuity at the endpoints,
i.e., they set $\lambda_1 = \mu_1 := 1$. Later, this approach was also used by Lu~\cite{Lu13}. In~\cite{RM13}, the same idea was used to simplify the $G^{3,3}$ case, and the authors noted that such approach leads to a system of linear equations.

Now, we generalize the above-described approach for any $k,\,l$ such that $-1 \leq k,l \leq 3$.
If $k \geq 2$, we set $\lambda_1 := 1$, which implies $C^{1}$-continuity at $t=0$ and
consequently, $G^{k,l}$ constraints become $C^{1,q}/G^{k,l}$ constraints, where $q \in \{-,1\}$.
Similarly, when $l \geq 2$, we set $\mu_1 := 1$, which implies $C^{1}$-continuity at $t=1$ and
consequently, $G^{k,l}$ constraints become $C^{p,1}/G^{k,l}$ constraints, where  $p \in \{-,1\}$.

Notice that in the cases of $k=2,3$ or $l=2,3$, the above-described method leads to the linear system~\eqref{Eq:GeoSystem-mod} and the error~\eqref{Eq:Distance2} is a quadratic function of the continuity parameters.
However, in the cases of $k = 1$ or $l = 1$, there is no guarantee that the solution satisfies $\lambda_1 > 0$ or $\mu_1 > 0$, respectively. In the case of solution with nonpositive values of these parameters, our choice is to solve a quadratic programming problem, subject to the constraints with prescribed positive lower bounds for the parameters (see~\eqref{Eq:PositiveCons}).

Observe that this approach uses no simplifying assumptions for $k, l \leq 1$.

Let us denote the above-described approach to Problem~\ref{P:GProblem} as \textit{$C^{p,q}/G^{k,l}$-constrained multi-degree reduction of B\'ezier curves}.

\section{Algorithms}\label{Sec:Alg}

In this section, we show the details of implementation of the proposed method of $G^{k,l}$-constrained multi-degree reduction of B\'ezier curves.
Moreover, we give a short description of $C^{p,q}/G^{k,l}$-constrained multi-degree reduction algorithm.

\subsection{$G^{k,l}$-constrained multi-degree reduction algorithm}\label{SubSec:GAlg}

Now, we give the method of solving Problem~\ref{P:GProblem}, summarized in the following two-phase algorithm.

Phase A of the algorithm consists in finding values of the parameters $\lambda_i$ and $\mu_j$ to minimize
the error~\eqref{Eq:Distance2}, which---by the results given in Theorem~\ref{Th:Problem0}---
depends only on these parameters. The idea is based on solving the quadratic or nonlinear programming problem (see~\S\ref{SubSec:NLPGeo}).
Notice that when $k,l < 1$, we can compute the $\phi_{ij}$ coefficients (cf.~\eqref{Eq:psi}) and omit the remaining steps of Phase A, since there are no continuity parameters to determine. During Phase B, we use the results of Theorem~\ref{Th:Problem0} and the obtained values of continuity parameters to compute control points $r_0, r_1,\dots, r_m$. Most of the known algorithms solve a system of normal equations, to get the inner control points of multi-degree reduced curve~\eqref{Eq:GReduced}. Such approach makes these points dependent on the inverse of a certain matrix. Our formulas do not require matrix inversion. What is more, the complexity of Phase B is $O(mn)$, which is significantly less than the cost of other known methods for this phase. The algorithm works for any $k$ and $l$ not exceeding 3.

\begin{algorithm}\,[\texttt{$G^{k,l}$-constrained multi-degree reduction}]\label{A:AlgorithmG}
	\ \\[0.5ex]
	\noindent \texttt{Data}: $\alpha,\,\beta$ -- parameters of the inner product
	\eqref{Eq:inprod};\\[0.5ex]
	\noindent \hphantom{\texttt{Data}:} $n$,  $p_0, p_1,\ldots, p_n$ -- degree and the control points of the B\'{e}zier curve~\eqref{Eq:GOriginal};\\[0.5ex]	
	\noindent \hphantom{\texttt{Data}:} $m$ -- degree of the reduced B\'{e}zier curve
	\eqref{Eq:GReduced};\\[0.5ex]
	\noindent \hphantom{\texttt{Data}:} $k,\,l$ -- orders of the $G$-continuity at the endpoints
	of the curve \eqref{Eq:GReduced};\\[0.5ex]
	\noindent \hphantom{\texttt{Data}:} $d_0,\,d_1$ -- lower bounds for the parameters $\lambda_1$ and $\mu_1$, respectively (cf. \S\ref{SubSec:NLPGeo}).\\[0.5ex]	
\noindent \texttt{Assumptions}: $n>m>0$;\; $-1\leq k,l\leq 3$;\; $k+l < m-1$;\; $d_0, d_1 > 0$;\; $\alpha, \beta > -1$.\\[0.5ex]
\noindent \texttt{Result}: control points $r_0, r_1,\ldots, r_m$ of the $G^{k,l}$-constrained multi-degree reduced B\'{e}zier curve~\eqref{Eq:GReduced}.\\

\noindent \textbf{Phase A}~\begin{description}
\itemsep4pt
\item[\texttt{Step I}] Compute $\{\phi_{ij}\}$ ($i=k+1,k+2,\ldots,m-l-1;\ j =0,1,\ldots,n$) by \cite[Algorithm 4.2]{LW11} and formula~\eqref{Eq:psi} (see Remark~\ref{R:AuxAlg}).
\item[\texttt{Step II}] Check if the remaining steps of Phase A can be omitted:\\[1ex]
\textbf{If} ($k < 1$) and ($l < 1$) \textbf{then} go to \texttt{Step VI}.
\item[\texttt{Step III}] Compute
	$E^{(\alpha,\beta)}(\lambda_1,\lambda_2,\ldots,\lambda_k,\mu_1,\mu_2,\ldots,\mu_l)$ by~\eqref{Eq:Distance2}.
\item[\texttt{Step IV}] Determine set $c$ of constraints:\\[1ex]
$c := \left\{\lambda_1 \geq d_0, \ \mu_1 \geq d_1\right\}$;\\[1ex]
\textbf{If} ($k < 1$) \textbf{then} $c := c \setminus \left\{\lambda_1 \geq d_0\right\}$;\\[1ex]
\textbf{If} ($l < 1$) \textbf{then} $c := c \setminus \left\{\mu_1 \geq d_1\right\}$.
\item[\texttt{Step V}] \mbox{}\\[1ex]
\textbf{If} ($k > 1$ or $l > 1$) \textbf{then}
\begin{itemize}
\item[]  obtain $\lambda_1,\lambda_2,\ldots,\lambda_k$, and $\mu_1,\mu_2,\ldots,\mu_l$ by solving the nonlinear programming problem of minimizing
                               the error~\eqref{Eq:Distance2}, subject to the constraints $c$;\\
\hspace*{-0.68cm} \textbf{else}
\item[]  obtain $\lambda_1,\lambda_2,\ldots,\lambda_k$, and $\mu_1,\mu_2,\ldots,\mu_l$ by solving the quadratic programming problem of minimizing the error~\eqref{Eq:Distance2}, subject to the constraints $c$.
\end{itemize}
\end{description}
\noindent \textbf{Phase B}
\begin{description}
\itemsep4pt
\item[\texttt{Step VI}] Compute
\vspace{1ex}
\begin{enumerate}
\itemsep3pt
\item $r_0,r_1,\ldots, r_k$ by~\eqref{Eq:GeoCondsGk1}--\eqref{Eq:GeoCondsGk3};

\item $r_{m-l},r_{m-l+1},\ldots, r_m$ by~\eqref{Eq:GeoCondsGl1}--\eqref{Eq:GeoCondsGl3}.
\end{enumerate}
\item[\texttt{Step VII}] Compute  $\upsilon_0,\upsilon_1,\ldots, \upsilon_n$ by~\eqref{Eq:Coeffsups}.	
\item[\texttt{Step VIII}] Compute  $r_{k+1},r_{k+2},\ldots, r_{m-l-1}$ by~\eqref{Eq:Solution0}.
\item[\texttt{Step IX}] Return the solution, i.e., the control points $r_0,r_1,\ldots, r_m$ of the
	reduced B\'{e}zier curve  \eqref{Eq:GReduced}.	
\end{description}
\end{algorithm}

\subsection{$C^{p,q}/G^{k,l}$-constrained multi-degree reduction algorithm}\label{SubSec:CGAlg}

Now, let us give the outline of the two-phase $C^{p,q}/G^{k,l}$-constrained multi-degree reduction algorithm.

Phase A of the algorithm implements the ideas discussed in~\S\ref{SubSec:CGeo}, therefore, it solves the system of linear equations~\eqref{Eq:GeoSystem-mod} to determine values of the continuity parameters. In the case of solution with nonpositive values of $\lambda_1$ or $\mu_1$, which can happen when $k = 1$ or $l = 1$, the algorithm solves a quadratic programming problem, subject to the constraints with prescribed positive lower bounds for the parameters  (see~\eqref{Eq:PositiveCons}). An example of a resulting B\'ezier curve that does not satisfy the positive condition can be found in
\cite[Fig. 1(a)]{Lu13}. We performed more than $40$ different tests (results of some of them are available in the next section). None of them caused such problem. Phase B is the same as for Algorithm~\ref{A:AlgorithmG}. The algorithm works for any $k$ and $l$ not exceeding 3.
For details, see our implementation in  Maple{\small \texttrademark}13 available on the website \url{http://www.ii.uni.wroc.pl/~pgo/GDegRed.mws}.

Obviously, Algorithm~\ref{A:AlgorithmG} costs more, but also produces more accurate results, since for the
$C^{p,q}/G^{k,l}$-constrained approach we additionally assume that $\lambda_1 = 1$ when $k > 1$, and
$\mu_1 = 1$ when $l > 1$.

\section{Examples}\label{Sec:Ex}
	
This section provides of the application of our $G^{k,l}$-constrained and $C^{p,q}/G^{k,l}$-constrained multi-degree reduction algorithms.
In each case, we give the least squares error $E^{(\alpha,\beta)}_2 := \sqrt{E^{(\alpha,\beta)}}$
and the maximum error
$$
E_{\infty} := \max_{t \in D_N} ||P_n(t) - R_m(t)|| \approx \max_{t \in [0,1]} ||P_n(t) - R_m(t)||,
$$
where $D_N := \left\{0, 1/N, 2/N,\ldots, 1\right\}$ with $N := 500$.

In our experiments, we consider the ``natural'' choices for the values of parameters $\alpha$, $\beta$, i.e., $(\alpha, \beta) \in \left\{(0,0), \left(\frac{1}{2},\frac{1}{2}\right), \left(-\frac{1}{2},\frac{1}{2}\right), \left(\frac{1}{2},-\frac{1}{2}\right), \left(-\frac{1}{2},-\frac{1}{2}\right)\right\}$, and set the lower bounds $d_0, d_1$ of $\lambda_1, \mu_1$ to $10^{-4}$ (see~\eqref{Eq:PositiveCons}).

\noindent Taking into account the different types of continuity constraints, we compare the following cases:
\begin{itemize}\setlength{\itemindent}{0.35cm}
\item[(i)]  $C^{k,l}$-constrained case (see~\eqref{Eq:CConstraints}), which can be solved by using Theorem~\ref{Th:Problem0};
\item[(ii)] $C^{p,q}/G^{k,l}$-constrained case, solved by the algorithm described in~\S\ref{SubSec:CGAlg};
\item[(iii)] $G^{k,l}$-constrained case, solved by Algorithm~\ref{A:AlgorithmG}.
\end{itemize}

Results of the experiments have been obtained on a computer with \texttt{Intel Core i5-3337U 1.8GHz} processor and \texttt{8GB} of \texttt{RAM}, using $32$-digit arithmetic.  Maple{\small \texttrademark}13 worksheet containing implementation of the algorithms and tests is available on the website \url{http://www.ii.uni.wroc.pl/~pgo/GDegRed.mws}. We use  Maple{\small \texttrademark} \texttt{fsolve} procedure, in the $C^{p,q}/G^{k,l}$ case, to solve the system of linear equations, and \texttt{QPSolve}, \texttt{NLPSolve} procedures, to solve the quadratic and nonlinear programming problems, respectively. \texttt{QPSolve} uses the iterative active-set method, and for \texttt{NLPSolve} we select \texttt{sqp} method.
Initial points for both procedures correspond to the values of continuity parameters in the $C^{k,l}$ case.

\begin{example}\label{Ex:alpha}
First, let us consider degree eleven B\'{e}zier curve which is an outline of the font ``alpha'' (for the control points, see~\cite[Example 6.1]{WL09}).
The results of multi-degree reduction are given in Table~\ref{Tab:1}. Figs.~\ref{Fig:1a} and~\ref{Fig:1b} illustrate two of the considered cases.
One can see, that when it comes to minimizing $E_{\infty}$ error, usually a good choice is $\alpha = \beta = -\frac{1}{2}$.
As expected, solution to the $G^{k,l}$ case is the most accurate, while $C^{p,q}/G^{k,l}$ approach gives less precise results.
$C^{k,l}$ conditions seem to be too restrictive, especially for $k$ or $l$ exceeding $2$.

\begin{figure}[H]
\captionsetup{margin=0pt, font={scriptsize}}
\begin{center}
\setlength{\tabcolsep}{0mm}
\begin{tabular}{c}
\subfloat[]{\label{Fig:1a}\includegraphics[width=0.48\textwidth]{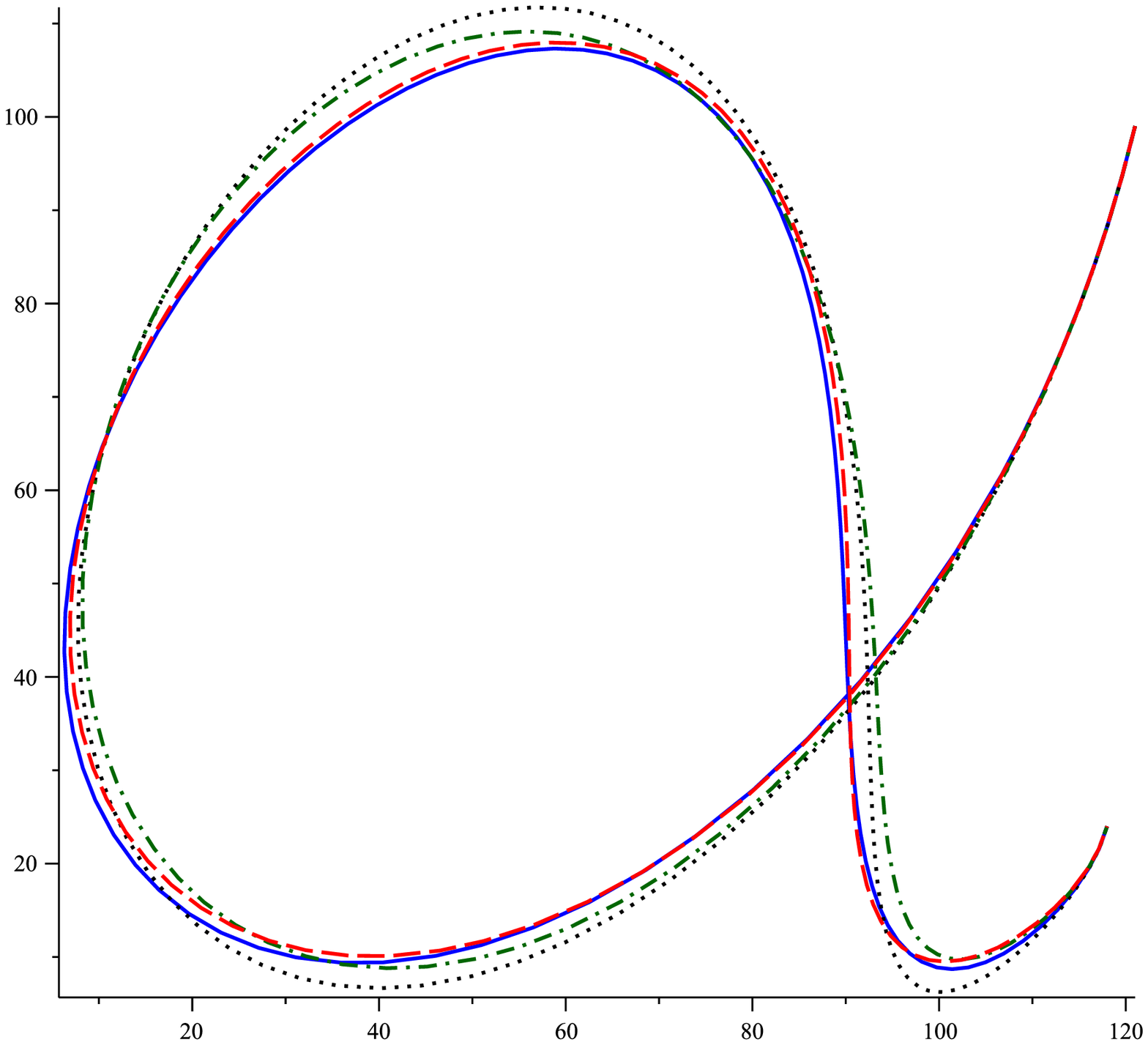}}
\subfloat[]{\label{Fig:1b}\includegraphics[width=0.52\textwidth]{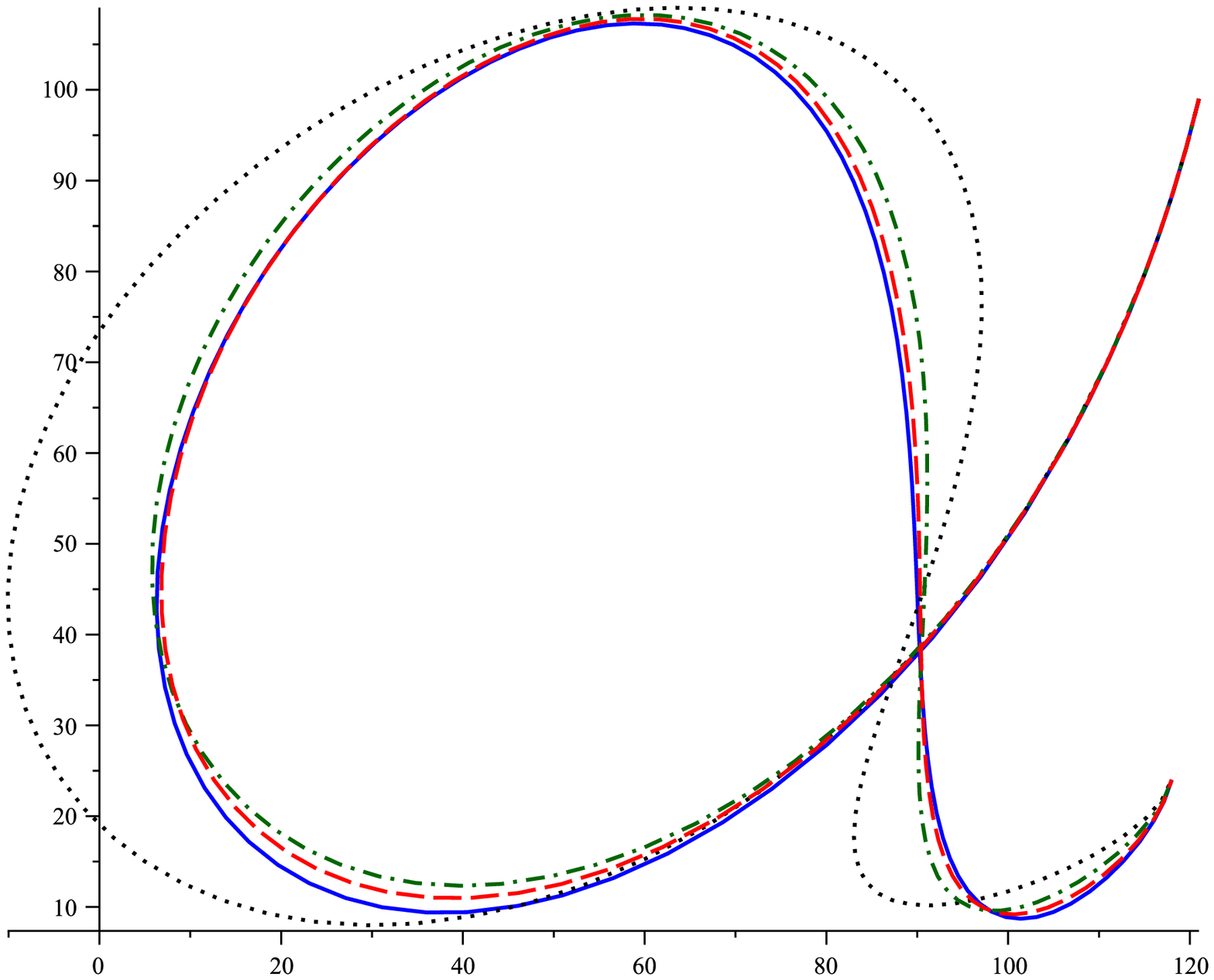}}
\end{tabular}
    \caption{Multi-degree reduction of degree eleven curve (blue solid line) to degree seven
     with $C^{k,l}$ (black dotted line), $C^{p,q}/G^{k,l}$ (green dash-dotted line)
    and $G^{k,l}$ (red dashed line) continuity constraints; parameters: (a) $\alpha = \beta = -\frac{1}{2}$,
    $p = q = 1$, $k = l = 2$, and  (b) $\alpha = \beta = -\frac{1}{2}$, $p = 1$, $q = -$, $k = 3$ and $l = 1$.}
\end{center}
\end{figure}

\begin{table}[H]
\captionsetup{margin=0pt, font={scriptsize}}
\centering
\ra{1.3}
\scalebox{0.98}{
\begin{tabular}{@{}rrrrrrrlcclcclcc@{}}
\toprule  \multicolumn{7}{c}{Parameters} & \phantom{abc} & \multicolumn{2}{c}{$C^{k,l}$ solution}
& \phantom{abc} & \multicolumn{2}{c}{$C^{p,q}/G^{k,l}$ solution}
& \phantom{abc} & \multicolumn{2}{c}{$G^{k,l}$ solution}
\\ \cmidrule{1-7} \cmidrule{9-10} \cmidrule{12-13} \cmidrule{15-16} $m$ & $k$ & $l$ & $p$ & $q$ & $\alpha$ & $\beta$ & \phantom{abc} &
$E_{2}^{(\alpha,\beta)}$ & $E_{\infty}$  & \phantom{abc} & $E_{2}^{(\alpha,\beta)}$ & $E_{\infty}$  & \phantom{abc} &
$E_{2}^{(\alpha,\beta)}$ & $E_{\infty}$
\\ \midrule
$7$ &  $2$ & $2$ & $1$& $1$ & $0$ & $0$ & &$3.73$ &$5.97$ & &$2.83$ &$5.27$ & &$0.93$ &$2.26$\\
  &    &    & &    &$-\frac{1}{2}$ & $-\frac{1}{2}$ & &$5.75$ &$5.83$ & &$4.40$ &$5.14$ & &$1.67$ &$1.91$\\
  &    &    & &    &$-\frac{1}{2}$ &  $\frac{1}{2}$ & &$3.83$ &$7.53$ & &$2.62$ &$6.42$ & &$0.79$ &$2.91$\\
  &    &    & &    &$\frac{1}{2}$ & $-\frac{1}{2}$ & &$3.83$ &$7.69$ & &$3.18$ &$5.18$ & &$1.26$ &$2.19$\\
  &    &    & &    &$\frac{1}{2}$ &  $\frac{1}{2}$ & &$2.43$ &$6.10$ & &$1.83$ &$5.40$ & &$0.53$ &$2.54$\\
\midrule
$7$ & $3$ & $1$ & $1$& $-$ & $0$ & $0$ & &$9.13$ &$16.24$ & &$2.51$ &$5.11$ & &$1.02$ &$2.41$\\
  &    &    & &    &$-\frac{1}{2}$ & $-\frac{1}{2}$ & &$13.97$ &$16.72$ & &$4.07$ &$4.95$ & &$1.81$ &$2.07$\\
  &    &    & &    &$-\frac{1}{2}$ &  $\frac{1}{2}$ & &$9.41$ &$19.51$ & &$2.18$ &$6.38$ & &$0.75$ &$3.11$\\
  &    &    & &    &$\frac{1}{2}$ & $-\frac{1}{2}$ & &$9.17$ &$18.79$ & &$2.98$ &$4.91$ & &$1.45$ &$1.99$\\
  &    &    & &    &$\frac{1}{2}$ &  $\frac{1}{2}$ & &$6.00$ &$15.82$ & &$1.56$ &$5.25$ & &$0.59$ &$2.69$\\
\end{tabular}}
\caption{Least squares error and maximum error in multi-degree reduction of degree eleven B\'{e}zier ``alpha'' curve.}
\label{Tab:1}
\end{table}


\end{example}


\begin{example}\label{Ex:heart}
Let us apply the algorithms to degree thirteen B\'{e}zier ``heart'' curve  (for the control points, see~\cite[Appendix~B]{RM13}) and consider the case of
$k = l = 2$. The results of experiments are given in Table~\ref{Tab:2}.
Notice that the case of $\alpha = \beta = 0$ was also considered in~\cite[\S5.2]{RM13} and \cite[Example 4]{ZWY13}.
As in~\cite{ZWY13}, we can clearly see that the solution to the $G^{2,2}$ case, in this paper obtained by Algorithm~\ref{A:AlgorithmG}, is more accurate than the result given by the approach proposed in~\cite{RM13}, which leads to the $C^{1,1}/G^{2,2}$ case (the same as for the algorithm discussed in~\S\ref{SubSec:CGAlg}).
As our approach considers different weight functions, it can be seen that the best choice to minimize $E_{\infty}$ error is $\alpha = \beta = -\frac{1}{2}$. Fig.~\ref{Fig:2} presents $\alpha = \beta = 0$ case.

Now, we focus on the running times. We have implemented $G^{1,1}$, $C^{1,1}/G^{2,2}$ and $C^{1,1}/G^{3,3}$-constrained methods from \cite{RM13}, $G^{1,1}$ and $C^{1,1}/G^{2,2}$-constrained methods from \cite{Lu13} as well as $G^{2,1}$ and $G^{2,2}$-constrained methods from \cite{ZWY13}. The methods of Rababah and Mann and of Lu solve the same problem and give the same results as our $C^{p,q}/G^{k,l}$-constrained method (see~\S\ref{SubSec:CGAlg}). In Table~\ref{Tab:3}, we compare the running times of these algorithms. Clearly, our approach is the fastest one. For the comparison of the $G^{k,l}$-constrained algorithms, see Table~\ref{Tab:4}. Notice that, in some cases, our $G^{k,l}$-constrained approach is slightly faster than the methods from~\cite{ZWY13}. We use Maple{\small \texttrademark} \texttt{fsolve} procedure to solve the cubic equation~\cite[(23)]{ZWY13} associated with the $G^{2,1}$-constrained case. The implementation of $G^{2,2}$-constrained method from~\cite{ZWY13} requires the \textit{unconstrained nonlinear programming solver}. According to our experiments, the \textit{nonlinear simplex method} (\texttt{NLPSolve} command with option \texttt{method = nonlinearsimplex} and the initial point $\lambda = \eta = 1$) is the fastest solver available in Maple{\small \texttrademark}13. Therefore, we use this solver for the purpose of the comparison. It is worth mentioning that Zhou et al. have omitted the constraints~\eqref{Eq:PositiveCons}. Consequently, in some rare cases, the resulting curve may not preserve the original tangent directions at the endpoints. To avoid this issue, one can implement the improvements proposed by Lu~\cite{Lu15}.

\begin{figure}[h]
\begin{center}
\captionsetup{margin=0pt, font={scriptsize}}
\includegraphics[width=145mm]{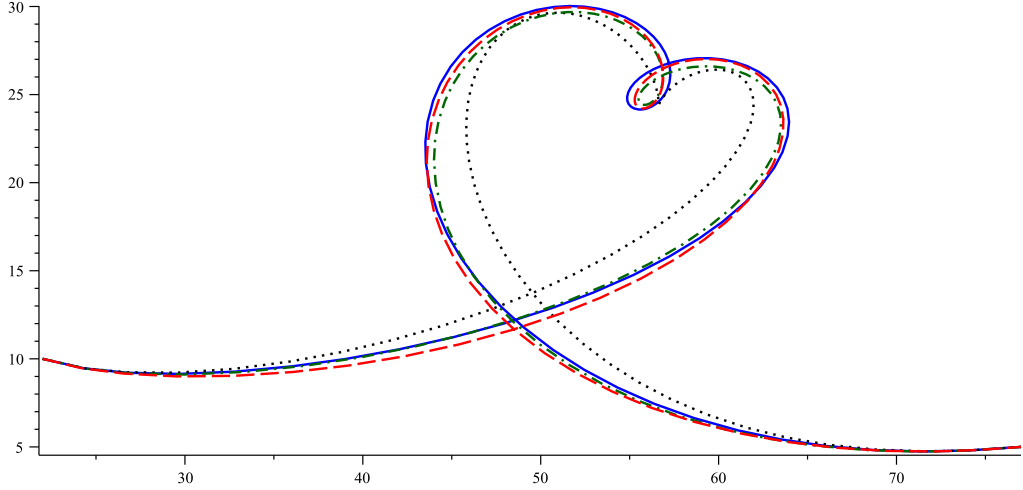}
\caption{Multi-degree reduction of degree thirteen curve (blue solid line) to degree eight
     with $C^{2,2}$ (black dotted line), $C^{1,1}/G^{2,2}$ (green dash-dotted line)
    and $G^{2,2}$ (red dashed line) continuity constraints; parameters: $\alpha = \beta = 0$.}
\label{Fig:2}
\end{center}
\end{figure}

\begin{table}[h]
\captionsetup{margin=0pt, font={scriptsize}}
\centering
\ra{1.3}
\scalebox{0.98}{
\begin{tabular}{@{}rrrlcclcclcc@{}}
\toprule \multicolumn{3}{c}{Parameters} & \phantom{abc} & \multicolumn{2}{c}{$C^{2,2}$ solution}
& \phantom{abc} & \multicolumn{2}{c}{$C^{1,1}/G^{2,2}$ solution}
& \phantom{abc} & \multicolumn{2}{c}{$G^{2,2}$ solution}
\\ \cmidrule{1-3} \cmidrule{5-6} \cmidrule{8-9} \cmidrule{11-12} $m$ & $\alpha$ & $\beta$ & \phantom{abc} &
$E_{2}^{(\alpha,\beta)}$ & $E_{\infty}$  & \phantom{abc} & $E_{2}^{(\alpha,\beta)}$ & $E_{\infty}$  & \phantom{abc} &
$E_{2}^{(\alpha,\beta)}$ & $E_{\infty}$
\\ \midrule
$8$ & $0$ & $0$ & &$1.52$ &$2.52$ & &$0.64$ &$1.12$ & &$0.36$ &$0.71$\\
  &   $-\frac{1}{2}$ & $-\frac{1}{2}$ & &$2.37$ &$2.39$ & &$1.05$ &$1.00$ & &$0.62$ &$0.53$\\
  &   $-\frac{1}{2}$ &  $\frac{1}{2}$ & &$1.58$ &$3.34$ & &$0.64$ &$1.55$ & &$0.42$ &$0.94$\\
  &   $\frac{1}{2}$ & $-\frac{1}{2}$ & &$1.52$ &$3.48$ & &$0.74$ &$1.31$ & &$0.37$ &$0.89$\\
  &   $\frac{1}{2}$ &  $\frac{1}{2}$ & &$0.98$ &$2.64$ & &$0.39$ &$1.24$ & &$0.21$ &$0.90$\\
\end{tabular}}
\caption{Least squares error and maximum error in multi-degree reduction of degree thirteen B\'{e}zier ``heart'' curve.}
\label{Tab:2}
\end{table}

\begin{table}[h]
\captionsetup{margin=0pt, font={scriptsize}}
\centering
\ra{1.3}
\scalebox{0.98}{
\begin{tabular}{@{}rrrrrlccc@{}}
\toprule \multicolumn{5}{c}{Parameters} & \phantom{abc} & \multicolumn{2}{c}{Running times [ms]}
\\ \cmidrule{1-5} \cmidrule{7-9}  $m$ & $k$ & $l$ & $p$ & $q$ & \phantom{abc} &
Our $C^{p,q}/G^{k,l}$ method & Rababah and Mann \cite{RM13} & Lu \cite{Lu13}
\\ \midrule
$8$ & $1$ & $1$ & $-$ & $-$ & & $32$ &$62$ & $63$\\
$10$ & $1$ & $1$ & $-$ & $-$ & & $31$ &$94$ & $63$\\
$12$ & $1$ & $1$ & $-$ & $-$ & & $47$ &$94$ & $62$\\
\midrule
$8$ & $2$ & $2$ & $1$ & $1$ & & $16$ &$63$ & $47$\\
$10$ & $2$ & $2$ & $1$ & $1$ & & $31$ &$63$ & $62$\\
$12$ & $2$ & $2$ & $1$ & $1$ & & $31$ &$78$ & $78$\\
\midrule
$8$ & $3$ & $3$ & $1$ & $1$ & & $15$ &$78$ & ---\\
$10$ & $3$ & $3$ & $1$ & $1$ & & $31$ &$109$ & ---\\
$12$ & $3$ & $3$ & $1$ & $1$ & & $63$ &$110$ & ---\\
\end{tabular}}
\caption{Running times of $C^{p,q}/G^{k,l}$-constrained multi-degree reduction of degree thirteen B\'{e}zier ``heart'' curve; parameters: $\alpha = \beta = 0$.}
\label{Tab:3}
\end{table}

\newpage
\begin{table}[h]
\captionsetup{margin=0pt, font={scriptsize}}
\centering
\ra{1.3}
\scalebox{0.98}{
\begin{tabular}{@{}cccccc@{}}
\toprule \multicolumn{3}{c}{Parameters} & \phantom{abc} & \multicolumn{2}{c}{Running times [ms]}
\\ \cmidrule{1-3} \cmidrule{5-6}  $m$ & $k$ & $l$ & \phantom{abc} &
Our $G^{k,l}$ method & Zhou et al.~\cite{ZWY13}
\\ \midrule
$8$ & $2$ & $1$ && $92$ & $108$ \\
$10$ & $2$ & $1$ && $121$ & $137$ \\
$12$ & $2$ & $1$ && $168$ & $166$ \\
\midrule
$8$ & $2$ & $2$ && $153$ & $204$\\
$10$ & $2$ & $2$ && $298$& $248$\\
$12$ & $2$ & $2$ && $290$& $292$\\
\end{tabular}}
\caption{Running times of $G^{k,l}$-constrained multi-degree reduction of degree thirteen B\'{e}zier ``heart'' curve; parameters: $\alpha = \beta = 0$.}
\label{Tab:4}
\end{table}

\end{example}

\section{Conclusions}\label{Sec:Conc}

In this paper, we propose efficient methods of solving the problems of $G^{k,l}$-constrained and $C^{p,q}/G^{k,l}$-constrained multi-degree reduction of B\'{e}zier curves with respect to the least squares norm. We give two-phase algorithms of solving these problems.\\
\indent The first phase of the algorithms consists in finding values of the geometric continuity parameters to minimize
the error~\eqref{Eq:Distance0}. In the case of $G^{k,l}$-constrained problem, we solve the
quadratic or nonlinear programming problem to obtain these values.
For $C^{p,q}/G^{k,l}$-constrained case, we use some simplifying assumptions, i.e., we impose constraints of
$C^{1}$-continuity at $t=0$ when $k > 1$, and at $t=1$ when $l > 1$. Therefore, by fixing some of the parameters, this approach leads to
the system of linear equations~\eqref{Eq:GeoSystem-mod}.
Assuming that $-1 \leq k,l \leq 3$ and including the hybrid cases, there are $37$ continuity cases which require computation of the continuity parameters. Those variants of the problem differ, and we have not proven that in each case a unique solution exists.

During the second phase, which is the same for both approaches, we use the properties of constrained dual Bernstein
basis polynomials to compute control points of the multi-degree reduced curve. The complexity of this phase
is $O(mn)$, where $n$ and $m$ are the degrees of the input and output curves, respectively.
This is significantly less than complexity of other algorithms. Moreover, our approach avoids matrix inversion.

As expected, solution to $G^{k,l}$-constrained problem is the most accurate,
while the one obtained by $C^{p,q}/G^{k,l}$-constrained multi-degree reduction is less precise.
$C^{k,l}$ conditions tend to be too restrictive, especially for $k$ or $l$ exceeding $2$.
Comparison of running times of our $C^{p,q}/G^{k,l}$-constrained approach with analogous methods from~\cite{Lu13,RM13} shows advantage of our algorithm in practice. Furthermore, the experiments show that our $G^{k,l}$-constrained approach is comparable to the methods of~\cite{ZWY13}, even slightly faster in some cases.

\bibliographystyle{elsart-num-sort}
\bibliography{Gdegred}

\begin{thebibliography}{10}
\expandafter\ifx\csname url\endcsname\relax
  \def\url#1{\texttt{#1}}\fi
\expandafter\ifx\csname urlprefix\endcsname\relax\def\urlprefix{URL }\fi

\bibitem{BG97}
J.~F. Bonnans, J.~C. Gilbert, C.~Lemarechal, C.~A. Sagastiz{\'a}bal, Numerical
  Optimization: Theoretical and Practical Aspects, Second Edition,
  Springer-Verlag, Berlin Heidelberg, 1997.

\bibitem{CW02}
G.~Chen, G.~Wang, Optimal multi-degree reduction of {B}\'{e}zier curves with
  constraints of endpoints continuity, Computer Aided Geometric Design 19
  (2002) 365--377.

\bibitem{Eck95}
M.~Eck, Least squares degree reduction of {B}\'{e}zier curves, Computer-Aided
  Design 27 (1995) 845--851.

\bibitem{Far02}
G.~E. Farin, Curves and Surfaces for Computer-Aided Geometric Design. A
  Practical Guide, Fifth Edition, Academic Press, Boston, 2002.

\bibitem{LW11}
S.~Lewanowicz, P.~Wo\'{z}ny, Multi-degree reduction of tensor product
  {B}\'{e}zier surfaces with general boundary constraints, Applied Mathematics
  and Computation 217 (2011) 4596--4611.

\bibitem{Lu13}
L.~Lu, Explicit ${G}^2$-constrained degree reduction of {B}\'{e}zier curves by
  quadratic optimization, Journal of Computational and Applied Mathematics 253
  (2013) 80--88.

\bibitem{Lu14}
L.~Lu, An explicit method for ${G}^3$ merging of two {B}\'ezier curves, Journal
  of Computational and Applied Mathematics 260 (2014) 421--433.

\bibitem{Lu15}
L.~Lu, Some improvements on optimal multi-degree reduction of {B}\'ezier curves
  with geometric constraints, Computer-Aided Design 59 (2015) 39--42.

\bibitem{LW06}
L.~Lu, G.~Wang, Optimal multi-degree reduction of {B}\'{e}zier curves with
  ${G}^2$-continuity, Computer Aided Geometric Design 23 (2006) 673--683.

\bibitem{WL07}
L.~Lu, G.~Wang, A quadratic programming method for optimal degree reduction of
  {B}\'{e}zier curves with ${G}^1$-continuity, Journal of Zhejiang University
  SCIENCE A 8 (2007) 1657--1662.

\bibitem{RM11}
A.~Rababah, S.~Mann, Iterative process for ${G}^2$ {m}ulti-degree reduction of
  {B}\'{e}zier curves, Applied Mathematics and Computation 217 (2011)
  8126--8133.

\bibitem{RM13}
A.~Rababah, S.~Mann, Linear methods for ${G}^1$, ${G}^2$, and
  ${G}^3$-{M}ulti-degree reduction of {B}\'{e}zier curves, Computer-Aided
  Design 45 (2013) 405--414.

\bibitem{Sun05}
H.~Sunwoo, Matrix representation for multi-degree reduction of {B}\'{e}zier
  curves, Computer Aided Geometric Design 22 (2005) 261--273.

\bibitem{SL04}
H.~Sunwoo, N.~Lee, A unified matrix representation for degree reduction of
  {B}\'{e}zier curves, Computer Aided Geometric Design 21 (2004) 151--164.

\bibitem{WL09}
P.~Wo\'{z}ny, S.~Lewanowicz, Multi-degree reduction of {B}\'{e}zier curves with
  constraints, using dual {B}ernstein basis polynomials, Computer Aided
  Geometric Design 26 (2009) 566--579.

\bibitem{ZW10}
L.~Zhou, G.~Wang, Matrix representation for optimal multi-degree reduction of
  {B}\'ezier curves with ${G}^1$ constraints, Journal of Computer Aided Design
  \& Computer Graphics 22 (2010) 735--740.

\bibitem{ZWY13}
L.~Zhou, Y.~Wei, Y.~Yao, Optimal multi-degree reduction of {B}\'{e}zier curves
  with geometric constraints, Computer-Aided Design 49 (2014) 18--27.

\end{thebibliography}
\nocite{*}

\end{document}